\documentclass[runningheads]{llncs}
\usepackage[T1]{fontenc}
\usepackage{graphicx}
\usepackage{hyperref}

\usepackage{color}

\usepackage[usenames,dvipsnames]{xcolor}

\usepackage{amsmath,latexsym,amssymb,amsfonts,amscd,stmaryrd,textcomp}
\usepackage{nicefrac}
\usepackage{xfrac}
\usepackage{subcaption}
\usepackage{textcomp}
\usepackage{tikz}
\usetikzlibrary{%
intersections,%
calc,%
fit,%
tikzmark,%
decorations.pathreplacing,%
decorations.pathmorphing,%
decorations.markings,%
patterns,%
arrows.meta,%
shapes%
}

\usepackage[linesnumbered,noend,ruled,vlined,resetcount]{algorithm2e}

\usepackage{cleveref}
\usepackage{marvosym}

\newcommand{\RR}{\mathbb{R}}
\newcommand{\RA}{\bar{\mathbb{Q}}}
\newcommand{\QQ}{\mathbb{Q}}
\newcommand{\NN}{\mathbb{N}}
\newcommand{\PP}[1]{\QQ[x_1,{\ldots},x_{#1}]}

\newcommand{\res}[3]{%
    \textup{res}_{x_{#3}}[#1,#2]%
}
\newcommand{\disc}[2]{%
    \textup{disc}_{x_{#2}}[#1]%
}

\newcommand{\coeff}[2]{%
    \textup{coeff}_{x_{#2}}[#1]%
}

\newcommand{\ldcf}[2]{%
    \textup{ldcf}_{x_{#2}}[#1]%
}

\newcommand{\dg}[2]{%
    \textup{deg}_{x_{#2}}[#1]%
}

\newcommand{\rex}[3]{%
    \textup{root}_{#3}[#1,#2]%
}

\newcommand{\ires}[1]{%
    \textup{IRE}(#1)%
}

\newcommand{\key}[1]{%
    \emph{#1}%
}

\newcommand{\apxscc}{%
    \texttt{apx-scc}%
}

\newcommand{\smt}{{\footnotesize SMT}\xspace}
\newcommand{\smtlib}{{\footnotesize SMT-LIB}\xspace}
\newcommand{\smtrat}{{\footnotesize SMT-RAT}\xspace}
\newcommand{\nra}{{\footnotesize NRA}\xspace}
\newcommand{\qfnra}{{\footnotesize QF-NRA}\xspace}
\newcommand{\cad}{{\footnotesize CAD}\xspace}
\newcommand{\nlsat}{{\footnotesize NLSAT}\xspace}
\newcommand{\mcsat}{{\footnotesize MCSAT}\xspace}

\definecolor{aux}{rgb}{0.99,0.78,0.07}
\definecolor{auxG}{rgb}{0.6,0.6,0.61}
\definecolor{auxLG}{rgb}{0.85,0.85,0.86}
\definecolor{auxOne}{rgb}{0,0,0.86}
\definecolor{auxTwo}{rgb}{0,0,0.61}
\definecolor{auxThree}{rgb}{0,0,0.31}
\colorlet{pOne}{blue!80!black}
\colorlet{pTwo}{green!70!black}
\colorlet{pThree}{teal}

\newcommand{\polyP}[1]{
    \draw[#1] (0,0) circle (1);
}

\newcommand{\polyQ}[1]{
    \draw[domain=-1.5:1.5, samples=50, #1] plot ({\x},{1.5*\x*\x*\x})
}

\newcommand{\polyR}[1]{
    \draw[domain=-1.5:1.3, samples=50, #1] plot ({\x},{0.5-0.5*\x+(1/(256*(\x-1.4)*(\x-1.4)*(\x-1.4)))})
}

\newcommand{\coordinateBox}{
    \draw[black, very thin] (-1.3,-1.15) rectangle (1.3,1.3);
    \node (x1) at (1,-1.3) {$x_1$};
    \node (x2) at (-1.45,1) {$x_2$};
}

\newcommand{\setupExample}{
    \begin{scope}
        \clip (-1.3,-1.15) rectangle (1.3,1.3);
  
        \polyP{ultra thick, pTwo, opacity=0.5};
        \polyQ{ultra thick, pThree, opacity=0.5};
        \polyR{ultra thick, pOne, opacity=0.5};

        \node[pTwo!70!black] (p2) at (0.4, 1.15) {$p_2$};
        \node[pThree] (p3) at (1.1,1.15) {$p_3$};
        \node[pOne!90!black] (p1) at (-0.9,1.15) {$p_1$};
      \end{scope}
}

\newcommand{\highlighted}[2]{
    \begin{scope}
        \clip (#1, -1.3) rectangle (#2, 1.3);
        \polyP{line width=2pt, pTwo, opacity=1};
        \polyQ{line width=2pt, pThree, opacity=1};
        \polyR{line width=2pt, pOne, opacity=1};
    \end{scope}
}

\newcommand{\exactS}{
    \begin{scope}
        \clip (0,0) circle (1);
        \fill[gray!50, opacity=1, domain=-0.76:0.54]
        (-0.76,-1.3) -- (-0.76,-0.65) -- plot ({\x},{1.5*\x*\x*\x}) -- (0.54,0.23) -- (0.54,-1.3) --cycle;
        \node[gray!50!black] (S) at (0,-0.5) {$S$};
    \end{scope}
}

\newcommand{\approxS}[3][magenta]{
    \begin{scope}
        \clip (0,0) circle (1);
        \fill[#1, opacity=0.33] (#2,-1.3) rectangle #3;
      \end{scope}
}

\newcommand{\projection}[3][dashed]{
    \draw[thick, #1, gray] (#2, #3) -- (#2, -1.15);
}

\newcommand{\sample}[3]{
    \node[blue, circle, fill, inner sep=1.5pt, label={[above left,yshift=-2]:#1}] (s) at (#2,#3) {};
}

\newcommand{\rootPoint}[2]{
    \node[diamond,fill,inner sep=1.5pt,label={#1}] at #2 {};
}

\newcommand{\rootLabel}[3][]{
    \node[label={[#1, below left=0]:#2}] at #3 {};
}

\newcommand{\rootLabelUp}[3][]{
    \node[label={[#1, above left=0]:#2}] at #3 {};
}

\newcommand{\symInt}[4][]{
    \draw[decorate, decoration={brace,amplitude=5pt,mirror,raise=2pt}, thick] #3 -- #4 node [#1, midway] {#2};
}

\begin{document}
\title{More is Less: Adding Polynomials for Faster Explanations in NLSAT}
\titlerunning{Adding Polynomials for Faster NLSAT Explanations}
%
\author{Valentin Promies\inst{1}\orcidID{0000-0002-3086-9976}\Letter \and
Jasper Nalbach\inst{1}\orcidID{0000-0002-2641-1380} \and
Erika Ábrahám\inst{1}\orcidID{0000-0002-5647-6134} \and
Paul Wagner\inst{1}}
\authorrunning{V. Promies et al.}
%
\institute{RWTH Aachen University, Aachen, Germany\\
\email{\{promies, nalbach, abraham\}@cs.rwth-aachen.de}}

\maketitle

\begin{abstract}
To check the satisfiability of (non-linear) real arithmetic formulas, modern satisfiability modulo theories (\smt) solving algorithms like \nlsat depend heavily on \emph{single cell construction}, the task of generalizing a sample point to a connected subset (cell) of $\RR^n$, that contains the sample and over which a given set of polynomials is sign-invariant.

In this paper, we propose to speed up the computation and simplify the representation of the resulting cell by dynamically extending the considered set of polynomials with further linear polynomials.
While this increases the total number of (smaller) cells generated throughout the algorithm, our experiments show that it can pay off when using suitable heuristics due to the interaction with Boolean reasoning.
\end{abstract}


\section{Introduction}\label{sec:intro}

\emph{Satisfiability checking} deals with the problem of deciding whether a first-order logic formula admits a solution. 
\emph{Satisfiability modulo theories (\smt)} solvers use specialized algorithms to tackle this problem for different theories.
While the targeted problems are generally hard (NP-complete for propositional logic, and even undecidable for  integer arithmetic), modern \smt solvers are highly efficient and widely used as integrated engines, e.g. for automated deduction \cite{sledgehammer,aprove}.

In this paper, we focus on the \emph{quantifier-free} fragment of \emph{non-linear real arithmetic (\nra)}, denoted as \emph{\qfnra}, whose formulas are Boolean combinations of polynomial constraints with rational coefficients and real-valued variables.
The \emph{cylindrical algebraic decomposition (\cad)} method \cite{DBLP:conf/automata/Collins75}, which is in general a quantifier elimination procedure for \nra,
was the first tractable technique for solving the satisfiability problem for \qfnra.
A \cad partitions the search space of the variables into a finite number of \emph{cells},
such that all polynomials in the input formula are sign-invariant - and thus the input formula is truth-invariant - within each \cad cell. Consequently, we can decide the satisfiability problem by checking one point from each cell.
Despite major improvements by, e.g., McCallum \cite{McCallum85,McCallum99}, Lazard \cite{Lazard1994}, and Brown \cite{Brown01a}, \cad still scales poorly, often due to expensive \emph{resultant} computations.

Inspired by \cad, one of the most successful \smt-solving techniques for \qfnra is \emph{\nlsat} \cite{JovanovicM12}, an instance of the \emph{model constructing satisfiability calculus (\mcsat)} \cite{MouraJ13} by Jovanic and de Moura.
\nlsat extends {\footnotesize DPLL+CDCL}-style propositional reasoning with a dual approach for the theory. The propositional part consists of deciding truth values for constraints, Boolean propagation, and Boolean conflict resolution. Dually, the theory part decides real values for theory variables, accompanied by theory propagation to assure that theory assignments evaluate constraints \emph{consistently} with the Boolean assignment, and theory conflict resolution.
Steps from both parts are interleaved, maintaining consistent partial assignments to guide each other towards a solution.

If all possible values for an unassigned theory variable would contradict the Boolean assignment of some constraint, then the current theory assignment cannot be extended consistently.
This situation is resolved using an \emph{explanation} backend, which generalizes the conflict's reason to a  lemma (called explanation) that not only excludes the current conflicting assignment but also further similar situations from the future search.
For example, given the formula $(x_1^2 + x_2^2 - 1 < 0\ \lor\ x_2 > 0)$, \nlsat could decide the first constraint to be true, and then assign $x_1 = 2$.
Now, no value of $x_2$ would satisfy $x_1^2 + x_2^2 - 1 < 0$, as it simplifies to $x_2^2 < -3$ under $x_1 = 2$. However, if $x_1^2 + x_2^2 - 1 < 0$ holds, then any value $x_1>1$ would require $x_2^2$ to be negative.
Thus we can generalize the value $x_1= 2$ to $x_1 > 1$ by learning the lemma $(\lnot(x_1^2 + x_2^2 - 1 < 0)\ \lor\ x_1 \leq 1)$.

Explanations should be efficiently computable and generalize as strongly as possible.
Note that the learnt clause may contain literals not present in the input formula (like $x_1 < 1$ above), thus the generalization technique is crucial for completeness.
\nlsat uses \cad techniques to generalize a sample point to a single cell around it.
However, \nlsat computes cells \emph{locally} w.r.t. a sample and a subset of the constraints, which offers potential for simpler computations and larger cells.
While improvements of Brown and Košta \cite{BrownK15}, Li and Xia \cite{samplecell}, and Nalbach et al. \cite{NALBACH2024102288} use this potential and avoid certain resultant computations,
single cell construction remains a major factor for the running time.

\paragraph*{Contributions.}
The cost of computing a resultant of two polynomials depends on their degree. If one of the polynomials is linear, the resultant is generally cheap to compute. In this paper, we dynamically insert linear polynomials during the cell construction from \cite{NALBACH2024102288}, effectively under-approximating the bounds on the cell.
This reduces the effort of the construction by replacing expensive resultant computations with simpler ones. It also affects the quality of the cell for \nlsat: The representation is simpler, but the cell covers less of the search space.
This paper follows up on the extended abstract \cite{PromiesNA24}, in which we briefly introduced this idea.
In particular, our contributions include the following:

\begin{itemize}
  \item We generalize the ideas from \cite{PromiesNA24} and provide a clear algorithmic formulation.
  \item We elaborate on the reasons for potential non-termination of \nlsat using the modified single cell construction by providing an example, and adapt the method to guarantee termination.
  \item We explore several variants of our method, including different ways of constructing additional polynomials and criteria for when to insert them.
  \item We provide an extensive experimental evaluation of these variants, using our implementation in the \smt solver \smtrat \cite{corzilius_smt-rat_2015,noauthor_smt-rat_nodate}.
\end{itemize}

\paragraph*{Structure.} In \Cref{sec:prelimiaries} we recall some background, including the levelwise single cell construction from \cite{NALBACH2024102288}, which we adapt for under-approximation in \Cref{sec:apx}.
We explore variants of our approach in \Cref{sec:variants} and discuss experiments in \Cref{sec:experiments}. In \Cref{sec:conclusion} we conclude with an outlook on future work.

\paragraph*{Related Work.}
Some incomplete but fast explanation backends for \nlsat are not based on \cad, but on Fourier-Motzkin variable elimination \cite{FMSCC}, virtual substitution \cite{virtualS}, or interval constraint propagation \cite{KremerPhd}.

The general idea of using linear approximations or abstractions for \qfnra has been explored before.
For example, the \texttt{ksmt} calculus \cite{ksmt} transforms each formula into a set of linear arithmetic clauses and a set of non-linear constraints and then incrementally constructs a model, mainly performing linear reasoning.
Partial assignments which falsify a non-linear constraint are generalized to conjunctions of linear constraints, using local linear approximations of the non-linear functions.
Incremental linearization \cite{inclin} computes a linear abstraction of the input formula by replacing all multiplications with uninterpreted functions, possibly allowing to derive unsatisfiability by purely linear reasoning.
The abstraction is refined incrementally by adding linear arithmetic axioms for the individual multiplications.
Neither of these methods is complete, though \texttt{ksmt} was shown to be $\delta$-complete.

\section{Preliminaries}\label{sec:prelimiaries}
We assume that the reader has some basic knowledge about multivariate polynomials, logic and \smt solving. For an introduction, we refer to, e.g., \cite{Barrett2018,polysCox}.

Let $\NN$, $\QQ$, and $\RR$ be the sets of natural (incl. 0), rational, respectively real numbers.
For $k \in \NN$, let $[k] := \{1,{\ldots},k\}$; for $r\in\RR^k$ and $i\in[k]$ let $r_i$ be the $i$th entry in $r$, $r_{[i]}:=(r_1,{\ldots},r_i)$, and $r_{[0]}=()$.
For the extent of this paper, we fix some $n \in \mathbb{N} \setminus \{0\}$ and ordered real-valued variables  $x_1 \prec {\ldots} \prec x_n$.

\paragraph*{Polynomials.}
For $i \in [n]$, let $\PP{i}$ be the set of all polynomials in $x_1, \ldots, x_i$ with rational coefficients
(for $i=0$, this is $\QQ$).
We can write any $p \in \PP{i}$ as a univariate polynomial $p = c_d x_i^d + c_{d-1} x_i^{d-1} + \ldots + c_1 x_i + c_0$ in $x_i$ with either $d=0$ or $c_d\neq 0$, with \key{degree} $\dg{p}{i} := d$, \key{coefficients} $\coeff{p}{i} := \{c_0,\ldots,c_d\} \subset \PP{i-1}$, and \key{leading coefficient} $\ldcf{p}{i} := c_d $.
Given $r \in \RR^i$, we write $p(r)$ for the evaluation $p(r_1,\ldots, r_i) \in \RR$.
Given $r \in \RR^{i-1}$ and $r'\in\RR$, let $p(r,x_i) \in \RR[x_i]$ result from $p$ by substituting $r_1,\ldots,r_{j-1}$ for $x_1,\ldots, x_{j-1}$, and we write $p(r,r')$ for $p(r_1,\ldots,r_{j-1},r')$.

\paragraph*{Real Roots.}
Let $p \in \PP{i}$. A \key{(real) root} of $p$ is a point $r \in \RR^i$ so that $p(r) = 0$;
the \key{variety} of $p$ is the set of its roots.
The roots of univariate polynomials build the set of \key{(real) algebraic numbers} $\RA := \{r \in \RR \mid \exists q \in \QQ[x]. q(r)=0\}$. 
Given $r \in \RA^{i-1}$, one can compute $\textup{realRoots}(p,r) := \{r_{j} \in \RR \mid p(r,r_{j}) = 0\}$, i.e. the \key{roots of $p$ over $r$}.
If $\textup{realRoots}(p,r) = \RR$, then we say that $p$ is \key{nullified over $r$}.
Otherwise, $\textup{realRoots}(p,r)$ is a finite set of algebraic numbers.

Let $j \in \NN$. An \key{indexed root expression} $\rex{p}{j}{x_i} : \RR^{i-1} \rightarrow \RR \cup \{\bot\}$ maps each $r \in \RR^{i-1}$ to the $j$-th root of $p$ over $r$ if it exists, and to $\bot$ otherwise:
\[
    \rex{p}{j}{x_i}(r) := \left.\begin{cases}
        \bot & \text{if } \textup{realRoots}(p,r) = \RR \text{ or } j > |\textup{realRoots}(p,r)|, \text{ and else}\\
        z_j & \text{where } \textup{realRoots}(p,r) = \{z_1, \ldots, z_k\},  z_1 < \ldots < z_k.
    \end{cases}\right.
\]
We refer to the polynomial $p$ of an indexed root expression $\xi = \rex{p}{j}{x_i}$ by $\xi.p$, and we say that the level of $\xi$ is $i$.
The set of indexed root expressions of level $i$ is $\ires{i}$.
Given $P \subseteq \PP{i}$ and $r \in \RA^{i-1}$, one can compute the \key{indexed root expressions defined over $r$}: $\textup{irExp}(P,r) := \{\xi \in \ires{i} \mid \xi.p \in P \text{ and } \xi(r) \neq \bot\}$.

The \key{resultant} of two polynomials $p,q \in \PP{i}$ w.r.t. $x_i$ is a polynomial $\res{p}{q}{i} \in \PP{i-1}$, such that for all $r \in \RR^{i-1}$ it holds: if there is $r'\in\RR$ with $p(r,r')=0 = q(r,r')$, then $\res{p}{q}{i}(r)=0$.
The \key{discriminant} of $p$ is $\disc{p}{i} := \res{p}{p'}{i}$, where $p'$ is the derivative of $p$ w.r.t $x_i$.

\paragraph*{Formulas.}
A formula of the quantifier-free fragment of (non-linear) real arithmetic (\qfnra) is a Boolean combination of (polynomial) \key{constraints} of the form $p \sim 0$, with $p \in \PP{n}$ and $\sim \in \{<,>,=,\leq,\geq,\neq\}$.
An \key{extended constraint} has the form $x_i \sim \xi$, where $\xi \in \ires{i}$ is an indexed root expression.

\paragraph*{Cells.}
A \key{cell} is a non-empty connected set $S \subseteq \RR^i$ for some $i \in [n]$.
We call $S$ (semi-)\key{algebraic} if it is the solution set of a conjunction of constraints and extended constraints.
We call $p$ \key{sign-invariant} over $S$, if the sign of $p(r)$ is the same for all points $r\in S$ (i.e. $\forall r\in S.\;p(r) \sim 0$ for a fixed $\sim \in \{<,=,>\}$).
We call $S$ sign-invariant for $P \subset \PP{i}$, if all $p \in P$ are sign-invariant over $S$. 

\subsection{Levelwise Single Cell Construction}\label{sec:scc}
Given a constraint set $C$ and an assignment $s\in\RR^i$, if all extensions of $s$ evaluate some constraints from $C$ to false, then we say that $s$ is \emph{inconsistent} with $C$.
In \nlsat, if the theory assignment is inconsistent with the constraints $C$ defined to be true by the Boolean assignment, then we generalize $s$ to a cell $S\subseteq\RR^i$, whose points are all inconsistent with $C$. To do so, we derive a set $P\subset\PP{i}$ of projection polynomials, such that $s\in S$ and the sign-invariance of $P$ over $S$ assures that all points in $S$ are inconsistent with $C$.
The learned explanation is then $(\lnot C \lor \lnot \varphi_S)$, where $\varphi_S$ is a conjunction of extended constraints defining $S$.

\begin{definition}
    Given $i\in [n]$, $P \subset \PP{i}$ and $s \in \RR^i$, the problem of \emph{single cell construction} (SCC) is to compute a description of an algebraic cell $S \subseteq \RR^i$ so that $s \in S$ and all $p \in P$ are sign-invariant over $S$.
\end{definition}

\noindent We now recall the \emph{levelwise} SCC approach \cite{NALBACH2024102288} by Nalbach et al., which is summarized in \Cref{algo:levelwise}, and which we will modify in \Cref{sec:apx}.
\begin{algorithm}
    \caption{$\texttt{levelwise-scc}(P, s)$}\label{algo:levelwise}
    \SetKwInOut{Input}{Input}
    \SetKwInOut{Output}{Output}
    \DontPrintSemicolon
    \Input{A finite $P \subseteq \PP{i}$, and $s \in \RR^i$ (with $i \in [n]$).}
    \Output{A description $(I_1 \land \ldots \land I_i)$ of a sign-invariant cell for $P$ containing $s$}
    \For{$j=i,\ldots,1$}{
        $P_j := P \cap (\PP{j} \setminus \PP{j-1})$\; \label{algo:loop_start}
        \textbf{if} $\textup{realRoots}(p,s_{[j-1]}) = \RR$ for some $p \in P_j$ \textbf{then} \Return{FAIL}\;
        $\{\xi_1, \ldots, \xi_k\} := \textup{irExp}(P_j, s_{[j-1]})$ s.t. $\xi_1(s_{[j-1]}) \leq \ldots \leq \xi_k(s_{[j-1]})$\; \label{algo:levelwise-line-roots}
        \textbf{if} $k=0$ \hspace{10.4em} \textbf{then} $I_j := (true)$\;
        \textbf{else if} $s_j = \xi_{\ell}(s_{[j-1]})$ for some $\ell$ \textbf{then} $I_j := (x_j = \xi_{\ell})$\;
        \textbf{else if} $s_j > \xi_k(s_{[j-1]})$ \hspace{4.5em} \textbf{then} $I_j := (\xi_k < x_j)$\;
        \textbf{else if} $s_j < \xi_1(s_{[j-1]})$ \hspace{4.5em} \textbf{then} $I_j := (x_j < \xi_{1})$\;
        \textbf{else} $I_j := (\xi_{\ell} < x_j \land x_j < \xi_{\ell+1})$ for the $\ell$ with $\xi_\ell(s_{[j-1]}) < s_j < \xi_{\ell + 1}(s_{[j-1]})$\;\label{algo:loop_mid}
        \textbf{add} to $P$ discriminants and coefficients ensuring delineability\;\label{algo:levelwise-line-del}
        \textbf{add} to $P$ resultants ensuring sign-invariance\;\label{algo:levelwise-line-res}\label{algo:loop_end}
    }
    \Return{$(I_1 \land \ldots \land I_i)$}
\end{algorithm}

Let $i\in [n]$, $s \in \RR^i$, $P \subset \PP{i}$, and for $j\in[i]$ let $P_j$ be the polynomials from $P$ with largest variable $x_j$ (i.e. those containing $x_j$, but not $x_{j+1}, \ldots, x_i$).
For each dimension $j=i,{\ldots},1$, the algorithm determines a \emph{symbolic interval} $I_j$ of the form $(x_j = \xi)$, $(x_j < \xi)$, $(x_j > \xi)$, or $(\xi < x_j \land x_j < \xi')$ for some $\xi,\xi' \in \ires{j}$, bounding the value of $x_j$ w.r.t. the lower variables $x_1,\ldots,x_{j-1}$.
For all $r \in \RR^{j-1}$ with $\xi(r) \neq \bot \neq \xi'(r)$, $I_j$ defines a \emph{concrete interval} $I_j(r) \subseteq \RR$ which is $\{\xi(r)\}$, $(-\infty,\xi(r))$, $(\xi(r),\infty)$, or $(\xi(r),\xi'(r))$, respectively.

The final cell described by $I_1 \land \ldots \land I_i$ is \emph{locally cylindrical}, i.e., $I_1$ defines a concrete interval $S_1 \subseteq \RR$, and for $j = 2,\ldots,i$, the root expressions in $I_{j}$ are defined everywhere over $S_{j-1}$, and they specify the cell
\[S_j = \{(r,r') \in \RR^j \mid r \in S_{j-1} \land r' \in I_j(r)\}.\]

To determine $I_j$, we assign $x_1,\ldots,x_{j-1}$ to the underlying sample $s_{[j-1]}$, and compute $\textup{realRoots}(p, s_{[j-1]})$ for all $p \in P_j$. These roots witness the indexed root expressions in \Cref{algo:levelwise-line-roots}.
The greatest root below (or equal to) $s_j$ and the smallest root above (or equal to) $s_j$ provide the interval boundaries (if they do not exist, $-\infty$ and $\infty$ are used). Thus, the polynomials in $P_j$ are sign-invariant over $\{ s_{[j-1]} \} \times I_j(s_{[j-1]})$.

The idea is now that the to-be-constructed underlying cell $S_{j-1}\subseteq \RR^{j-1}$ will be a neighbourhood around $s_{[j-1]}$ over which the root expressions in $I_j$ define total continuous functions such that $S_j$ is a sign-invariant cell for $P_j$ containing $s_{[j]}$.
To obtain an underlying cell with the desired properties, the concepts of delineability and order-invariance (a strengthening of sign-invariance) are used:

\begin{definition}[Delineability \cite{DBLP:conf/automata/Collins75}]
    Let $j\in[n{-}1]$ and $S \subseteq \RR^j$ be a cell.
    A non-zero polynomial $p \in \PP{j+1}$ is \emph{delineable over $S$} if there exist $k \geq 0$ continuous functions $\theta_1, \ldots, \theta_k: S \rightarrow \RR$ and constants $m_1,\ldots m_k \in \NN$ such that for all $r \in S$ holds $\theta_1(r) < \ldots < \theta_k(r)$, $\textup{realRoots}(p,r) = \{\theta_1(r), \ldots, \theta_k(r)\}$, and for all $\ell=1,\ldots,k$ the multiplicity of the root $\theta_{\ell}(r)$ in $p(r,x_{j+1})$ is $m_{\ell}$.
\end{definition}
\begin{definition}[Order-invariance \cite{McCallum85}]
    The \emph{order of $p \in \PP{j}$ at $r \in \RR^{j}$}, denoted $ord(p,r)$, is the minimum $k$ so that some partial derivative of $p$ of total order $k$ does not evaluate to $0$ at $r$ (or $\infty$, if all evaluate to $0$).
    We call $p$ \emph{order-invariant} on $R \subseteq \RR^{j}$ if $ord(p,r) = ord(p,r')$ for all $r,r' \in R$.
\end{definition}

The indexed root expressions of $p$ determined in \Cref{algo:levelwise-line-roots} witness the $\theta$ functions. That these are well-defined continuous functions is assured by the delineability of $p$.
The method uses the fact that $p$ is delineable over the underlying cell $S_{j-1}$ if $\disc{p}{j}$ is order-invariant on $S_{j-1}$ and $\ldcf{p}{j}$ is sign-invariant over $S_{j-1}$; thus it adds these polynomials to $P$ and ensures their properties on the next level, thereby restricting $I_{j-1}$ and the levels below.

The method still has to ensure that no root function crosses the cell boundaries (the root expressions in $I_j$) over $S_{j-1}$, because this would imply a sign change of some polynomial within $S_j$.
For this purpose, we use that for any two $\xi,\xi' \in \textup{irExp}(P_j,s_{[j-1]})$ and $\sim\; \in \{<,=\}$ with $\xi(s_{[j-1]}) \sim \xi'(s_{[j-1]})$, it holds: If $\res{\xi.p}{\xi'.p}{j}$ is order-invariant on $S_{j-1}$, then $\xi(r) \sim \xi'(r)$ for all $r \in S_{j-1}$.

Since only intersections of roots with the cell boundaries are relevant, it suffices to maintain a \emph{partial} ordering of the root functions, ensured by certain resultants (\Cref{algo:levelwise-line-res}).
For example, if $\xi_1,\ldots,\xi_k$ are as in \Cref{algo:levelwise} and $I_j = (\xi_{\ell} < x_j \land x_j < \xi_{\ell+1})$, then we could add $\{\res{\xi_{\ell'}.p}{\xi_{\ell}.p}{j} \mid \ell' < \ell\} \cup \{\res{\xi_{\ell+1}.p}{\xi_{\ell'}.p}{j} \mid \ell + 1 < \ell'\}$, ensuring that $\xi_1,\ldots,\xi_{\ell-1}$ stay below (or equal to) $\xi_{\ell}$ and $\xi_{\ell+2},\ldots,\xi_k$ stay above (or equal to) $\xi_{\ell+1}$.
By exploiting transitivity, other partial orderings and thus other sets of resultants are also viable; this is a heuristic decision.
However, the resultant $\res{\xi_{\ell}.p}{\xi_{\ell+1}.p}{i}$ of the bounds is always added to ensure connectedness of $S_j$.

Note that this method fails if any of the encountered polynomials is nullified on the underlying sample, because then delineability cannot be ensured in the same way.
The method can detect this and return ``FAIL'', and a different, complete approach is used instead.
To further ensure that no polynomial $p \in P_j$ is nullified over any other point in $S_{j-1}$, some $c \in \coeff{p}{j}$ with $c(s_{[j-1]}) \neq 0$ is also added to $P$ in \Cref{algo:levelwise-line-del}.
After adding all required polynomials to $P$, if $j>1$, then the method proceeds with the construction of $I_{j-1}$ in the same way.
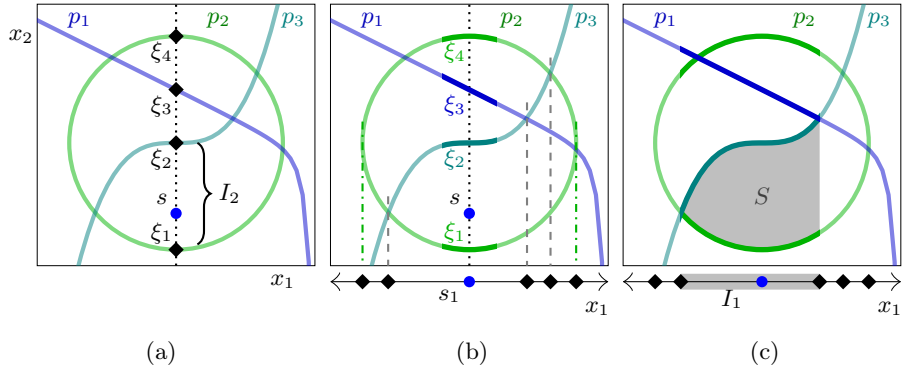
\begin{figure}[t]
    \begin{subfigure}{0.35\textwidth}
      \begin{tikzpicture}[scale=1.42]
        \coordinateBox
        \node at (1,-1.55) {\phantom{$x_1$}};
    
        \draw[black, thick, dotted] (0,-1.15) -- (0,1.3);
  
        \setupExample
    
        \rootPoint{[below left,xshift=2,yshift=-2]:$\xi_4$}{(0,1)}
        \rootPoint{[below left,xshift=2,yshift=-2]:$\xi_3$}{(0,0.5)}
        \rootPoint{[below left,xshift=2,yshift=-2]:$\xi_2$}{(0,0)}
        \rootPoint{[above left,xshift=2,yshift=-4]:$\xi_1$}{(0,-1)}
      
        \sample{$s$}{0}{-0.66}
        \symInt[xshift=0.5cm]{$I_2$}{(0.15,-0.95)}{(0.15,0)}
    
      \end{tikzpicture}
      \caption{}\label{fig:levelwise:a}
    \end{subfigure}
    \begin{subfigure}{0.31\textwidth}
      \begin{tikzpicture}[scale=1.42, >={Classical TikZ Rightarrow[scale=2]}]
        \draw[black, very thin] (-1.3,-1.15) rectangle (1.3,1.3);
        \node at (1.2,-1.55) {$x_1$};
        \draw[black, thick, dotted] (0,-1.15) -- (0,1.3);
        \setupExample
        \highlighted{-0.25}{0.25}
  
        \projection{-0.76}{-0.5}
        \projection{0.76}{0.8}
        \projection{0.54}{0.38}
  
        \draw[thick, dash dot, pTwo] (-1,0.2) -- (-1,-1.15);
        \draw[thick, dash dot, pTwo] (1,0.2) -- (1,-1.15);
  
        \rootLabelUp[pTwo,xshift=2,yshift=-4]{$\xi_1$}{(0,-1)}
        \rootLabel[pThree,xshift=2,yshift=-2]{$\xi_2$}{(0,0)}
        \rootLabel[pOne,xshift=2,yshift=-2]{$\xi_3$}{(0,0.5)}
        \rootLabel[pTwo,xshift=2,yshift=-2]{$\xi_4$}{(0,1)}
  
        \sample{$s$}{0}{-0.66}
  
        \draw[<->] (-1.3, -1.3) -- (1.3,-1.3);
        \node[blue, circle, fill, inner sep=1.5pt, label={[below left,yshift=-2]:$s_1$}] (s) at (0,-1.3) {};
        \rootPoint{}{(-0.76,-1.3)}
        \rootPoint{}{(0.54,-1.3)}
        \rootPoint{}{(0.76,-1.3)}
        \rootPoint{}{(-1,-1.3)}
        \rootPoint{}{(1,-1.3)}
      \end{tikzpicture}
      \caption{}\label{fig:levelwise:b}
  \end{subfigure}
    \begin{subfigure}{0.31\textwidth}
      \begin{tikzpicture}[scale=1.42, >={Classical TikZ Rightarrow[scale=2]}]
        \draw[black, very thin] (-1.3,-1.15) rectangle (1.3,1.3);
        \node (x1) at (1.2,-1.55) {$x_1$};
        \exactS
        \setupExample
        \highlighted{-0.76}{0.54}
        
        \draw[line width=6pt, gray!50] (-0.76,-1.3) -- (0.54,-1.3);
        \draw[<->] (-1.3, -1.3) -- (1.3,-1.3);
        \node[blue, circle, fill, inner sep=1.5pt, label={[below left,yshift=-2,xshift=-4]:$I_1$}] (s) at (0,-1.3) {};
        \rootPoint{}{(-0.76,-1.3)}
        \rootPoint{}{(0.54,-1.3)}
        \rootPoint{}{(0.76,-1.3)}
        \rootPoint{}{(-1,-1.3)}
        \rootPoint{}{(1,-1.3)}
  
      \end{tikzpicture}
      \caption{}\label{fig:levelwise:c}
  \end{subfigure}
    \caption{Illustration of the levelwise construction described in \Cref{ex:levelwise}.}\label{fig:levelwise}
  \end{figure}
\begin{example}\label{ex:levelwise}
    \Cref{fig:levelwise} illustrates an example with a given sample $s \in \mathbb{R}^2$ and polynomials $P = \{p_1, p_2, p_3\} \subset \mathbb{Q}[x_1,x_2]$.
    The line labelled with $p_1$ indicates the variety of $p_1$ i.e. those points $r\in \mathbb{R}^2$ with $p_1(r) = 0$, and similarly for $p_2, p_3$.
    
    We start at level $2$, where $P_2 = P$.
    At $x_1 = s_1$, there is one root of $p_2$ below $s_2$ and one root of each polynomial above $s_2$.
    The roots closest to $s_2$ define the symbolic interval $I_2 := (\xi_1 < x_2 \land x_2 < \xi_2)$ (\Cref{fig:levelwise:a}).
    To ensure correctness of this interval for all values of $x_1$ in the underlying cell (to be computed at level $1$), the discriminants and leading coefficients of $p_1, p_2, p_3$ are added to $P$ (dash-dotted lines in \Cref{fig:levelwise:b}).
    Moreover, adding $\res{p_3}{p_1}{2}$ and $\res{p_3}{p_2}{2}$ (dashed lines in \Cref{fig:levelwise:b}) ensures that none of the root functions cross the upper interval bound defined by $p_1$ over $I_1$.
    Note that the crossing of $\xi_3$ and $\xi_4$ is irrelevant, and the corresponding resultant of $p_1$ and $p_2$ is thus avoided.
    On level 1, we isolate the roots of these polynomials and use the closest to $s_1$ as interval boundaries, resulting in the shaded cell (\Cref{fig:levelwise:c}).
\end{example}
%
\section{Adding Polynomials to Avoid Expensive Resultants}\label{sec:apx}

The running time of SCC is dominated by discriminant and resultant computations.
Given $p,q \in \PP{n}$ with $d_j = \dg{p}{j}$, $e_j = \dg{q}{j}$ for all $j \in [n]$, the resultant of $p$ and $q$ requires $\mathcal{O}(d_n e_n)$ polynomial multiplications and, in the worst case, its degree w.r.t. any $x_j$ is $d_n e_j + d_j e_n$.
Given $P\subset \PP{n}$ with maximal degree $d$ in any variable, the degree and time complexity of resultants during SCC grows doubly exponential in worst-case ($d^{2^n}$), as resultants of resultants are computed iteratively.
The levelwise method already mitigates the effort for computing resultants, e.g. it avoids involving polynomials of high degree. However, some cannot be avoided, e.g. polynomials defining the bounds of an interval are always included in some resultant computations.

Our approach is as follows:
If a high-degree polynomial $p$ would define a bound of a symbolic interval $I_j$, then  we add a new \emph{linear} polynomial $p_* = (x_j - c)$ to $P$, whose root $c \in \QQ$ lies \emph{strictly} between that bound and the sample $s_j$.
Using this as a new, under-approximating bound for $I_j$ allows replacing expensive resultants of $p$ and some $q\in P$ by resultants of $p_*$ with $q$, which are simply computed by substituting $c$ for $x_j$ in $q$, and their degree is bounded by the one of $q$. 
The choice of $c$ ensures that (1) the resulting cell still contains the sample, (2) all other roots remain outside the cell, and (3) the underlying levels still generalize to some larger cell. Towards the latter, $c$ should not be equal to any polynomial's root, as the cell then would only generalize to a section $I_{j-1}=(x_{j-1} = \xi)$ on the level below as the resultant of that polynomial with $p_*$ would have a root at $s_{[j-1]}$.

\begin{figure}[t]
    \begin{subfigure}{0.35\textwidth}
        \begin{tikzpicture}[scale=1.42]
          \coordinateBox
          \node at (1,-1.55) {\phantom{$x_1$}};
      
          \draw[black, thick, dotted] (0,-1.15) -- (0,1.3);
    
          \setupExample
      
          \rootPoint{[below left,xshift=2,yshift=-2]:$\xi_4$}{(0,1)}
          \rootPoint{[above right,xshift=0,yshift=-4]:$\xi_3$}{(0,0.5)}
          \rootPoint{[above left,xshift=2,yshift=-4]:$\xi_2$}{(0,0)}
          \rootPoint{[above left,xshift=2,yshift=-4]:$\xi_1$}{(0,-1)}
        
          \sample{$s$}{0}{-0.66}
      
        \end{tikzpicture}
        \caption{}\label{fig:apx-a}
      \end{subfigure}
      \begin{subfigure}{0.31\textwidth}
        \begin{tikzpicture}[scale=1.42, >={Classical TikZ Rightarrow[scale=2]}]
          \draw[black, very thin] (-1.3,-1.15) rectangle (1.3,1.3);
          \node (x1) at (1.2,-1.55) {$x_1$};
          \draw[black, thick, dotted] (0,-1.15) -- (0,1.3);
          \setupExample

          \draw[ultra thick, magenta] (-1.3,-0.33) -- (1.3,-0.33) node[above,pos=0.06] {$p_*$};
          \node[magenta,fill,inner sep=2.5pt,label={[above right,xshift=2,yshift=-4,magenta]:$\xi_*$}] at (0,-0.33) {};
    
          \projection{-0.59}{-0.33}
          \projection{0.93}{-0.33}
          \projection{-0.93}{-0.33}
          \projection{1.16}{-0.33}
          \draw[thick, dash dot, pTwo] (-1,0.2) -- (-1,-1.15);
          \draw[thick, dash dot, pTwo] (1,0.2) -- (1,-1.15);

          \rootPoint{[below left,xshift=2,yshift=-2]:$\xi_4$}{(0,1)}
          \rootPoint{[above right,xshift=0,yshift=-4]:$\xi_3$}{(0,0.5)}
          \rootPoint{[above left,xshift=2,yshift=-4]:$\xi_2$}{(0,0)}
          \rootPoint{[above left,xshift=2,yshift=-4]:$\xi_1$}{(0,-1)}
    
          \symInt[xshift=0.5cm]{$I_2$}{(0.15,-0.95)}{(0.15,-0.35)}

          \sample{$s$}{0}{-0.66}

          \draw[<->] (-1.3, -1.3) -- (1.3,-1.3);
          \node[blue, circle, fill, inner sep=1.5pt, label={[below left,yshift=-2]:$s_1$}] (s) at (0,-1.3) {};
          \rootPoint{}{(-1,-1.3)}
          \rootPoint{}{(1,-1.3)}
          \rootPoint{}{(-0.93,-1.3)}
          \rootPoint{}{(-0.59,-1.3)}
          \rootPoint{}{(0.93,-1.3)}
          \rootPoint{}{(1.16,-1.3)}

        \end{tikzpicture}
        \caption{}\label{fig:apx-b}
    \end{subfigure}
      \begin{subfigure}{0.31\textwidth}
        \begin{tikzpicture}[scale=1.42, >={Classical TikZ Rightarrow[scale=2]}]
          \draw[black, very thin] (-1.3,-1.15) rectangle (1.3,1.3);
          \node (x1) at (1.2,-1.55) {$x_1$};

          \begin{scope}
            \clip (0,0) circle (1);
            \fill[pattern=north west lines, pattern color=gray, opacity=1, domain=-0.76:0.54]
            (-0.76,-1.3) -- (-0.76,-0.65) -- plot ({\x},{1.5*\x*\x*\x}) -- (0.54,0.23) -- (0.54,-1.3) --cycle;
            \node[gray!50!black] (S) at (0,-0.15) {$S$};
          \end{scope}

          \approxS{-0.59}{(0.93,-0.33)}
          \node[black] (S2) at (0,-0.66) {$S'$};

          \setupExample
          \highlighted{-0.59}{0.93}
          
          \draw[ultra thick, magenta] (-1.3,-0.33) -- (1.3,-0.33);          
          
          \draw[line width=6pt, magenta!50] (-0.59,-1.3) -- (0.93,-1.3);
          \draw[<->] (-1.3, -1.3) -- (1.3,-1.3);
          \node[blue, circle, fill, inner sep=1.5pt, label={[below right,yshift=-3,xshift=5]:$I_1$}] (s) at (0,-1.3) {};
          \rootPoint{}{(-1,-1.3)}
          \rootPoint{}{(1,-1.3)}
          \rootPoint{}{(-0.93,-1.3)}
          \rootPoint{}{(-0.59,-1.3)}
          \rootPoint{}{(0.93,-1.3)}
          \rootPoint{}{(1.16,-1.3)}
        \end{tikzpicture}
        \caption{}\label{fig:apx-c}
    \end{subfigure}
  \caption{%
    Approximated cell using a new linear polynomial
  }\label{fig:apx}
\end{figure}
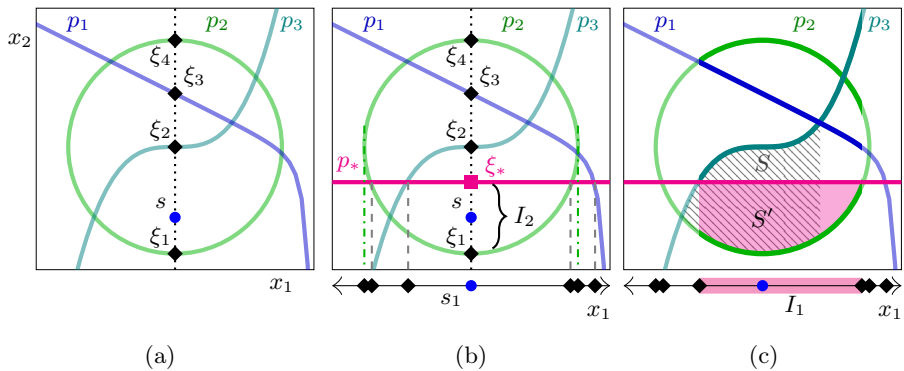

\begin{example}\label{ex:apx}
    In \Cref{ex:levelwise} (depicted in \Cref{fig:apx-a}), the levelwise method cannot avoid the resultants $\res{p_2}{p_3}{2}$ and $\res{p_3}{p_1}{2}$, which are expensive if $p_1$, $p_2$ and $p_2$ have high degree.
    Adding a linear polynomial $p_*$ with a root $\xi_*$ between $s$ and $\xi_2$, lets us use $\xi_*$ as upper bound of $I_2$, and it suffices to compute only the cheap resultants of $p_*$ with $p_1,p_2,p_3$ (like shown in \Cref{fig:apx-b}).
    \Cref{fig:apx-c} shows the resulting cell $S'$ (pink, shaded) and the original cell $S$ (hatched).
    Importantly, $p_1,p_2$ and $p_3$ are sign-invariant over $S'$.
    Note that we might also under-approximate the lower bound of $I_2$, or the bounds of $I_1$, leading to an even smaller cell, but also reducing computations.
\end{example}

We generalize this approach and modify \texttt{levelwise-scc}, so that on each level, it can dynamically extend the working set $P$ with arbitrary polynomials, resulting in our new method \apxscc{} shown in \Cref{algo:apx}, which adds \Cref{algo:apx-line}.
The method \texttt{apx-criteria} decides whether adding new polynomials is beneficial, by checking e.g. whether the symbolic interval $I_j$ would be defined by a polynomial with high degree.
If the criteria are fulfilled, \texttt{apx-polys} computes a set of new polynomials (called \emph{auxiliary} polynomials) that are added to $P$.
The auxiliary polynomials have roots at favourable positions, admitting an easier set of resultants to be computed. We discuss the different approaches for implementing \texttt{apx-criteria} and \texttt{apx-polys} in \Cref{sec:variants}.

\begin{algorithm}[b]
    \caption{$\texttt{apx-scc}(P, s)$}\label{algo:apx}
    \SetKwInOut{Input}{Input}
    \SetKwInOut{Output}{Output}
    \DontPrintSemicolon
    \Input{A finite $P \subseteq \PP{i}$ and $s \in \RR^i$ (with $i\in[n]$).}
    \Output{A description $(I_1,\ldots,I_i)$ of a sign-invariant cell for $P$ containing $s$}
    \For{$j=i,\ldots,1$}{
        $P_j := P \cap (\PP{j} \setminus \PP{j-1})$\;
        \textbf{if} $\textup{realRoots}(p,s_{[j-1]}) = \RR$ for some $p \in P_j$ \textbf{then} \Return{FAIL}\;
        \textbf{if} \texttt{apx-criteria(}$P_j,s_{[j]}$\texttt{)} \textbf{then} $P := P \cup \texttt{apx-polys(}P_j,s_{[j]}\texttt{)}$\;\label{algo:apx-line}
        \textbf{compute} $I_j$ as before (Lines \ref{algo:loop_start}-\ref{algo:loop_mid} in Algorithm \ref{algo:levelwise})\;
        \textbf{add} to $P$ polynomials ensuring delineability and sign-invariance (Lines \ref{algo:levelwise-line-del}-\ref{algo:loop_end} in Algorithm \ref{algo:levelwise})\;
    }
    \Return{$(I_1,\ldots,I_i)$}
\end{algorithm}

Adding auxiliary polynomials makes the maximal possible sign-invariant cell around the given sample point smaller, hence we compute some kind of \emph{under-approximation}.
However, as shown in \Cref{fig:apx-c}, the cell computed by \apxscc{} is not necessarily a subset of the cell computed by \texttt{levelwise-scc}, as strengthening the bounds of $I_j$ might allow weakening some bound of $I_{j-1}$.
In any case, both the original and the approximated cell ($S$ and $S'$) are subsets of the maximal sign-invariant cell $S_{max}$ for $P$.
While $S = S_{max}$ may hold, it always holds $S' \subsetneq S_{max}$.

Our modification has two main benefits for the usage in \nlsat: (1) Avoiding expensive resultant computations means that explanations can be computed much faster; and (2) during later computations, \nlsat needs to isolate the roots of the cell boundaries over further sample points, for checking whether a given sample lies in the excluded cell - the effort for these computations may be drastically reduced by polynomials of lower degree (or even degree $1$).
On the other hand, the under-approximated cells may lead to more cells generated throughout the search, and even lead to non-termination, as we will see in \Cref{sec:incompleteness}.

It is important to note that our modification does not eliminate the strong degree growth entirely, because the discriminants needed for delineability are also resultants.
Moreover, for sections $I_j = (x_j = \xi)$ we cannot compute meaningful approximations, forcing us to fall back to the default method.

\begin{theorem}[Correctness]
  Let $P \subset \PP{i}$ be finite and $s \in \RR^i$.
  If \texttt{apx-scc}$(P,s)$ yields the cell $S \subseteq \RR^i$, then $s \in S$ and $P$ is sign-invariant over $S$.
\end{theorem}
\begin{proof}
    The idea is that the original method could produce the same cell, when given an appropriately modified input.
    For each $j \in [i]$, let $Q_j \subset \PP{j}$ be the set of polynomials added by \apxscc{} on level $j$, and let $Q := Q_1 \cup \ldots \cup Q_i$.
    Consider the cell $S' \subseteq \RR^i$ computed by \texttt{levelwise-scc}$(P \cup Q, s)$.
    As that method is correct, $s \in S'$ and $P \cup Q$ is sign-invariant over $S'$.
    By definition of sign-invariance, this implies that $P$ is sign-invariant over $S'$.

    We show $S = S'$ to complete the proof:
    The polynomials in $Q_j$ do not impact the computations of the levels $j + 1, \ldots, i$, because $Q_j \cap P_k = \emptyset$ for each $k \neq j$.
    Thus, both \texttt{levelwise-scc} and \apxscc{} compute each level $j$ only based on $P \cup \bigcup_{k=j}^i Q_k$. As their computations do not differ (apart from adding $Q_j$), they compute exactly the same intervals and projections.\qed
\end{proof}

\subsection{Incompleteness}\label{sec:incompleteness}

Like \texttt{levelwise-scc}, our approach fails in the case of nullification and is thus incomplete as a stand-alone procedure.
However, one can detect nullification and resort to a complete construction for that cell.
More critically, the termination of \nlsat is no longer guaranteed when using \apxscc{} for explanations.

\begin{example}[Non-Termination]\label{ex:incompleteness}
    We continue \Cref{ex:apx}.
    After excluding the approximated cell from the search, \nlsat chooses a new value for $y$.
    This value can lie between the auxiliary cell boundary and the root of $p_3$, leading to the sample $s'$ as shown in \Cref{fig:incompleteness-a}.
    Since $s'$ and $s$ are in the same maximal sign-invariant cell for $p_1,p_2,p_3$, it leads to a conflict with the same constraints, and thus \apxscc{} is called with the same polynomials.
    When computing the new explanation, another auxiliary boundary is introduced between $s'$ and $p_3$, and this behaviour repeats, leading to a sample $s''$ as in \Cref{fig:incompleteness-b}.
    This can repeat indefinitely, and \nlsat will run into the same conflict over and over, without ever covering the entire search space.
    This behaviour cannot occur with the original construction, as is illustrated in \Cref{fig:incompleteness-c}.
\end{example}
\begin{figure}[t]
  \begin{subfigure}{0.35\textwidth}
    \begin{tikzpicture}[scale=1.425]
      \coordinateBox

      \draw[thick, dotted, gray] (0,-1.15) -- (0,1.3);
      
      \approxS{-0.59}{(0.93,-0.33)}

        \setupExample
        
        \draw[ultra thick, magenta] (-1.3,-0.33) -- (1.3,-0.33);  

      \node[blue, circle, fill, inner sep=1.5pt, label={[left, yshift=-2]:$s$}] at (0,-0.5) {};
      \node[blue, circle, fill, inner sep=1.5pt, label={[left, yshift=-2]:$s'$}] at (0,-0.20) {};
    \end{tikzpicture}
    \caption{}\label{fig:incompleteness-a}
  \end{subfigure}
  \begin{subfigure}{0.31\textwidth}
    \begin{tikzpicture}[scale=3.0875]
      \draw[black, very thin] (-0.6,-0.6) rectangle (0.6,0.5308);
      \node (x1) at (0.5,-0.675) {$x_1$};
      \begin{scope}
        \clip (-0.6,-0.6) rectangle (0.6,0.5308);

        \approxS{-0.59}{(0.93,-0.33)}
        \approxS{-0.45}{(1.3,-0.15)}
        \approxS{-0.3}{(0.9,-0.05)}
        
        \setupExample
        \draw[ultra thick, magenta] (-1.3,-0.33) -- (1.3,-0.33);

        \node[blue, circle, fill, inner sep=1.5pt, label={[left, yshift=-2]:$s$}] at (0,-0.5) {};
        \node[blue, circle, fill, inner sep=1.5pt, label={[left, yshift=-2]:$s'$}] at (0,-0.25) {};
        \node[blue, circle, fill, inner sep=1.5pt, label={[left, yshift=-2]:$s''$}] at (0,-0.1) {};
      \end{scope}
    \end{tikzpicture}
    \caption{}\label{fig:incompleteness-b}
\end{subfigure}
  \begin{subfigure}{0.31\textwidth}
    \begin{tikzpicture}[scale=1.425]
        \draw[black, very thin] (-1.3,-1.15) rectangle (1.3,1.3);
        \node (x1) at (1,-1.3) {$x_1$};
        \begin{scope}
            \clip (0,0) circle (1);
            \fill[gray!50, opacity=1, domain=-0.76:0.54]
            (-0.76,-1.3) -- (-0.76,-0.65) -- plot ({\x},{1.5*\x*\x*\x}) -- (0.54,0.23) -- (0.54,-1.3) --cycle;
        \end{scope}
        \setupExample
          \node[blue, circle, fill, inner sep=1.5pt, label={[left, yshift=-2]:$s$}] at (0,-0.5) {};
          \node[gray, circle, fill, inner sep=1.5pt, label={[gray,right, yshift=-2]:$s'$}] at (0,-0.25) {};
          \node[gray, circle, fill, inner sep=1.5pt, label={[gray,right, yshift=-2]:$s''$}] at (0,-0.1) {};
    \end{tikzpicture}
    \caption{}\label{fig:incompleteness-c}
\end{subfigure}
  \caption{%
    Non-termination of \nlsat with under-approximating cells, as described in \Cref{ex:incompleteness}. %
    Figure (b) zooms in on the area between $p_3$ and the approximation.
  }\label{fig:incompleteness}
\end{figure}
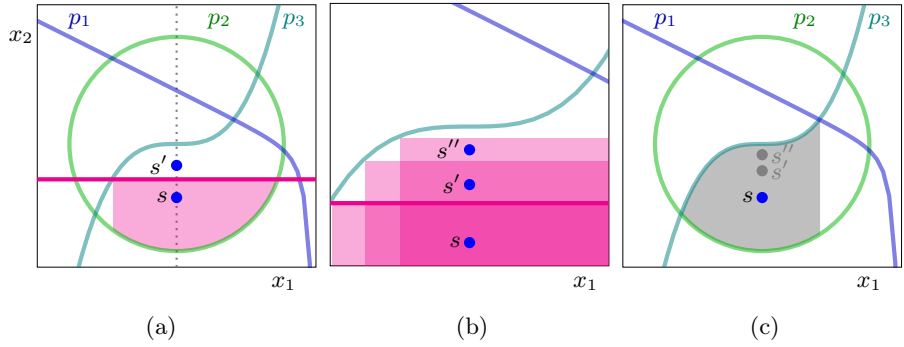
Note that \nlsat does not always run into this situation. It might also choose samples that are further away, thus escaping the conflict.
Thanks to the Boolean structure of the formula, the remaining search space might still be covered by approximated cells for other conflicts, which involve different sets of polynomials.

We can also express this with formal terms from \cite{JovanovicM12}: 
the termination of \nlsat relies on the fact that a \emph{finite basis} explanation function is used.
That is, for every input formula $\varphi$, there is a \emph{finite} set $B(\varphi)\subset \PP{n}$ such that in all possible runs of \nlsat, all explanations use only polynomials from $B$.
Since there are runs of \nlsat for which our explanation produces \emph{infinitely} many different literals, \apxscc{} does not yield a \emph{finite basis} explanation function.
\subsubsection*{Retaining Completeness.}
If we ensure that in every run of \nlsat, we add only finitely many auxiliary polynomials, then termination is guaranteed again.
The reasoning is that there will be a point after which the output of \apxscc{} is always equal to the original levelwise construction \texttt{levelwise-scc}.
Since this is a finite-basis explanation function, termination will be guaranteed.
\begin{lemma}[Termination]\label{lem:termination}
    A run of \nlsat using \apxscc{} terminates iff \texttt{apx-criteria(}$P_j, s_{[j]}$\texttt{)} returns \texttt{true} only finitely many times during that run.
\end{lemma}
To make use of this lemma, we provide additional information to \texttt{apx-criteria}, like the number $n_{cells}$ of so far approximated cells.
Then, \texttt{apx-criteria} could return \texttt{false} whenever $n_{cells}$ exceeds some fixed threshold, fulfilling the condition from \Cref{lem:termination} and thus implying termination.

However, the optimal threshold will vary depending on the input formula.
A more flexible approach is to gradually strengthen the criteria as the number of approximated cells increases.
For example, we only add a polynomial if a cell boundary is defined by some $p \in P$ with $\dg{j}{p} \geq c\cdot n_{cells} + d$, where $c,d \in \QQ_{\geq 0}$.

This avoids the behaviour from \Cref{ex:incompleteness}, since $c\cdot n_{cells} + d$ eventually exceeds the degree of the involved polynomials, but other cells with more expensive resultants are still approximated.
Importantly, this only ensures termination because the added polynomials are linear; otherwise, the degrees of polynomials derived in the construction could grow indefinitely, always fulfilling the criterion.

The guarantee may also be given by other criteria, e.g. using individual counters for the involved polynomials.

\section{Variants}\label{sec:variants}
We now present several instantiations for \texttt{apx-polys}  and \texttt{apx-criteria}. The first three methods approximate a cell boundary:
Given the sample $s_{[j]}$ and some $p \in P$, whose root $\xi$ would be a bound of $I_j$, we construct some $p_*$ with a root $c \in \QQ$ between $s_j$ and $b:=\xi(s_{[j-1]})$.

\subsubsection*{Simple Approach.}
At the beginning of \Cref{sec:apx}, we already introduced the idea of adding polynomials $x_j - c$ defining a constant bound $c$ on $x_j$. We now elaborate on the choice of $c$. 
Although choosing $c$ close to $b$ restricts $I_j$ less, it may shrink the underlying cell depending on the shape of $p$'s variety.
This can be observed in \Cref{fig:incompleteness}: the closer the approximate bound is to the actual bound, the smaller becomes $I_1$ (the interval for $x_1$).
It is thus not immediately clear what to choose. 
More importantly, this approach can produce numbers $c = \sfrac{num}{den}$ with large bit size $log(num) + log(den)$, causing significant overhead in operations like substituting $c$ into high-degree polynomials.
Therefore, we choose $c$ with a minimal bit size using a method based on the \emph{Stern-Brocot tree} \cite{brocot1861,stern1858}.

The following two approaches try to provide better approximations of the cell boundary, hoping to increase the cell's quality for \nlsat.

\subsubsection*{Taylor.}
We want to construct $p_*$ so that its gradient at its root $(s_{[j-1]},c)$ is equal to the gradient of $p$ at its root $r := (s_{[j-1]},b)$:
\[
    p_*\quad =\quad \frac{\partial p}{\partial x_j}(r)\cdot (x_j - c) + \sum\nolimits_{k \in [j-1]}\Big(\frac{\partial p}{\partial x_k}(r)\cdot (x_k - s_k)\Big)
\]
This is a slight modification of the \emph{first-order Taylor expansion} of $p$ at $r$, the difference being that the constant term $p(r)$ is left out as it is always zero, and that $(x_j - c)$ is used instead of $(x_j - r_j)$ in the first term.
This ensures that $p_*$ has its root at $(s_{[j-1]},c)$ instead of $r$.
Now, clearly 
\[p_*(s_{[j-1]},c) = 0\ \text{and}\ \frac{\partial p_*}{\partial x_k}(s_{[j-1]},c) = \frac{\partial p}{\partial x_k}(r) \text{ for all } k \in [j].\]
The idea is that the root functions of $p$ and $p_*$ will behave similarly around $s_{[j-1]}$, as illustrated in \Cref{fig:variants-a}.
The dashed line is the root of the tangent to $p_3$ at $(s_{j-1}, b)$, which is then shifted to pass through $c$.

Unfortunately, if some sample coordinate $s_k$ is \emph{irrational}, then (1) we cannot use the term $(x_k - s_k)$, since we require \emph{rational} coefficients, and (2) the gradients $\sfrac{\partial p}{\partial x_k}(r)$ are harder to compute and might also be irrational.
We tackle (1) by omitting the summand corresponding to $x_k$ for the $k$ where $s_k$ is irrational, and (2) by finding rational approximations of the gradients: Each (irrational) algebraic number can be isolated using an open interval containing a single root of its defining polynomial; this interval can be refined arbitrarily by bisection.
For $k \in [j]$, let $g_k \approx \sfrac{\partial p}{\partial x_k}(r)$ be a rational approximation, then we get
\[
    p_*\quad =\quad g_j\cdot (x_j - c) + \sum_{k \in [j-1] \text{ s.t. } s_k \in \QQ}\Big(g_k\cdot (x_k - s_k)\Big)
\]

Both mitigations harm the quality of the approximation, as some gradients are only approximately equal, or not counted in at all.
\subsubsection*{Piecewise Linear.}
Instead of approximating the boundary at a single point, we can use piecewise linear interpolation.
For this purpose, we determine an interval $D \subseteq \RR$ around $s_{j-1}$ so that $p$ is delineable over
\[
    \{s_{[j-2]}\} \times D = \{(s_{[j-2]},s') \in \RR^{j-1} \mid s' \in D\} \subseteq \RR^{j-1},
\]
i.e. $\disc{p}{j}$ and certain coefficients of $p$ are sign-invariant over that set.
We now know that $\xi$ is a total continuous function over $\{s_{[j-2]}\}\times D$, which is needed for deriving a meaningful interpolation.

We then choose $k \in \NN$ support points $d_1 < \ldots < d_k$ from $D$ such that $s_{j-1}\in \{d_1,\ldots,d_k\}$, and for each $\ell\in [k]$ we compute the value $\xi((s_{[j-2]},d_{\ell}))$ of the root function at that support point (in practice, we only isolate it in a rational interval).
Each of those values is under-approximated by choosing a value $d_{\ell}' \in \QQ$ ``close'' to the boundary $\xi((s_{[j-2]},d_{\ell}))$ such that $(s_{[j-2]},d_{\ell},d'_{\ell}) \in \RR^j$ is inside the cell.
We require that one support point $d_{\ell}$ is equal to $s_{j-1}$, because this guarantees that $(s_{[j-2]},d_{\ell},d'_{\ell})$ (for this particular $\ell$) is between the bound and the sample point.

The approximate bound then consists of $k-1$ pieces, which connect the points $\{(s_{[j-2]},d_\ell,d'_\ell) \in \RR^j \mid \ell \in [k]\}$ and which are defined by the roots of
\[p_*^{(\ell)} := (d_{\ell+1} - d_{\ell})(x_j - d'_{\ell}) - (d'_{\ell+1} - d'_{\ell})(x_{j-1}-d_{\ell}), \quad \ell \in [k-1].\]
That is, $p_*$ is a piecewise function so that for $r \in \RR^j$ holds $p_*(r) = p_*^{(\ell)}(r)$ if $d_{\ell} \leq r \leq d_{\ell+1}$ (or $\ell = k$ if $d_k \leq r$).
We can encode this function by a minimum over maximum of linear functions as described in \cite{xu2014,xu2016}, and thus derive an \qfnra formula. Further, we can adapt the subsequent cell construction to handle such compound interval bounds by using the techniques described in \cite{merging}: We compute the resultant of each $p_*^{(\ell)}$ with the polynomials below or above the interval and filter out roots of the resultants witnessing spurious intersections.
This approach is illustrated in \Cref{fig:variants-b} where the intersections with the dashed lines are filtered out.

There are two cases where we cannot apply this approach (and thus need to apply e.g. the simple approach):
(1) Like the Taylor approach, if $s_{j-1}$ is irrational, we cannot construct $p_*$, because $s_{j-1} = d_{\ell}$ for some $\ell \in [k]$, and thus we cannot use the term $x_{j-1} - d_{\ell}$.
(2) If $D = \{s_{j-1}\}$, i.e. $D$ contains only one point, we cannot choose more than one support point.
However, at least two points are needed to construct one piece of the approximate boundary.
This case happens e.g. when the discriminant of $p$ (the original boundary-defining polynomial) has a root $\xi_{disc}$ at $s_{[j-1]}$, which means that the next level will collapse to a section $I_{j-1} = (x_{j-1} = \xi_{disc})$ anyway and the behaviour of $p_*$ around the sample does not matter as much.

\begin{figure}[t]
  \begin{subfigure}{0.35\textwidth}
    \begin{tikzpicture}[scale=1.425]
      \coordinateBox
      \setupExample
      \begin{scope}
        \clip (0,0) circle (1);
        \fill[magenta, opacity=0.33] (-0.6,-1) -- (-0.6,-0.3) -- (0.8,0.1) -- (0.8,-1) --cycle; 
      \end{scope}
      \draw[thick, gray, dotted] (0.25,-1.15) -- (0.25,1.3);
      \node[blue, circle, fill, inner sep=1.5pt, label={[left, yshift=-2]:$s$}] at (0.25,-0.5) {};
      \draw[domain=-1.3:1.3, samples=10, dashed, thick, magenta] plot ({\x},{(9*\x-1.5)/32});
      \draw[domain=-1.3:1.3, samples=10, ultra thick, magenta] plot ({\x},{(9*\x-1.5)/32 - 0.1});
    \end{tikzpicture}
    \caption{}\label{fig:variants-a}
  \end{subfigure}
  \begin{subfigure}{0.31\textwidth}
    \begin{tikzpicture}[scale=1.425]
        \draw[black, very thin] (-1.3,-1.15) rectangle (1.3,1.3);
        \node (x1) at (1,-1.3) {$x_1$};
      \setupExample
      \begin{scope}
        \clip (0,0) circle (1);
        \fill[magenta, opacity=0.33] (-0.75,-1) -- (-0.75,-0.64) -- (-0.6,-0.474) -- (-0.3,-0.1905) -- (0,-0.15) -- (0.3,-0.1095) -- (0.6,0.174) -- (0.6,-1) --cycle;
      \end{scope}
      \begin{scope}
        \clip (-1.3,-1.15) rectangle (1.3,1.3);
        \draw[ultra thick, magenta] (-0.3,-0.1905) -- (0,-0.15) -- (0.3,-0.1095);
        \draw[domain=-1.3:1.3, samples=10, thick, magenta, dashed] plot ({\x},{\x*(0.15 - 0.1095)/0.3 + (-0.15)});
        \draw[domain=-1.3:1.3, samples=10, thick, magenta, dashed] plot ({\x},{\x*(0.474 - 0.1905)/0.3 + (0.474 - 2*0.1905)});
        \draw[domain=-1.3:1.3, samples=10, thick, magenta, dashed] plot ({\x},{\x*(0.174 + 0.1095)/0.3 + (-2*0.1095-0.174)});
        \draw[domain=-1.3:-0.3, samples=10, ultra thick, magenta] plot ({\x},{\x*(0.474 - 0.1905)/0.3 + (0.474 - 2*0.1905)});
        \draw[domain=0.3:1.3, samples=10, ultra thick, magenta] plot ({\x},{\x*(0.174 + 0.1095)/0.3 + (-2*0.1095-0.174)});
      \end{scope}
      \draw[thick, gray, dotted] (0,-1.15) -- (0,0);
      \draw[thick, gray, dotted] (0.3,-1.15) -- (0.3,0.05);
      \draw[thick, gray, dotted] (-0.3,-1.15) -- (-0.3,0);
      \draw[thick, gray, dotted] (0.6,-1.15) -- (0.6,0.35);
      \draw[thick, gray, dotted] (-0.6,-1.15) -- (-0.6,-0.25);
      \node at (-0.6,-1.3) {$d_1$};
      \node at (0,-1.3) {$\cdots$};
      \node at (0.6,-1.3) {$d_k$};
      \node[blue, circle, fill, inner sep=1.5pt, label={[right, yshift=-2]:$s$}] at (0,-0.5) {};
    \end{tikzpicture}
    \caption{}\label{fig:variants-b}
  \end{subfigure}
  \begin{subfigure}{0.31\textwidth}
    \begin{tikzpicture}[scale=1.425]
        \draw[black, very thin] (-1.3,-1.15) rectangle (1.3,1.3);
        \node (x1) at (1,-1.3) {$x_1$};
        \begin{scope}
            \clip (0,0) circle (1);
            \fill[magenta, opacity=0.33, domain=-0.76:0.47]
            (-0.76,-1.3) -- (-0.76,-0.65) -- plot ({\x},{1.5*\x*\x*\x}) -- (0.47,0.23) -- (0.47,-1.3) --cycle;
        \end{scope}
        \projection{0.47}{0.15}
        \projection{0.65}{0.15}
      \setupExample
      \draw[thick, gray, dotted] (0,-1.15) -- (0,1.3);
      \node[blue, circle, fill, inner sep=1.5pt, label={[left, yshift=-2]:$s$}] at (0,-0.5) {};
      \draw[ultra thick, magenta] (-1.3,0.15) -- (1.3,0.15);
    \end{tikzpicture}
    \caption{}\label{fig:variants-c}
  \end{subfigure}
  \caption{%
    Variants of the modified construction: (a) using Taylor expansion, (b) using piecewise linear bounds, and (c) adding roots outside the cell.
  }\label{fig:variants}
\end{figure}
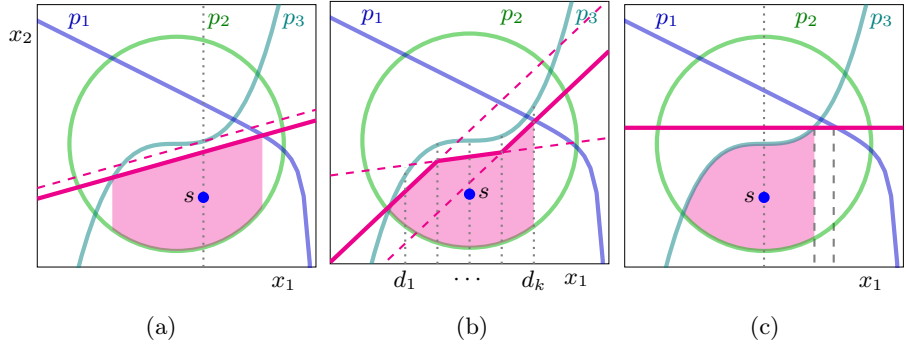
\subsubsection*{Roots Outside the Cell.}
We could also introduce polynomials $p_*$ with a root function $\xi_*$ between the cell boundary $\xi$ and some other root $\xi'$ outside the cell.
Instead of $\res{\xi.p}{\xi'.p}{j}$, we compute $\res{\xi.p}{p_*}{j}$ and $\res{p_*}{\xi'.p}{j}$ and then use transitivity, which may simplify the projection.
While this has the advantage of leaving the top level $I_i$ unchanged, the underlying cell may still be restricted by the new resultants.
This is illustrated in \Cref{fig:variants-c}.

\subsubsection*{Approximation Criteria.}
In addition to varying the way \emph{how} auxiliary polynomials are computed, we can also adapt \emph{when} they are computed, by using different instantiations of \texttt{apx-criteria}.

As our goal is to avoid expensive resultant computations, we are interested in quantities which influence the resultant complexity on the current level.
A simple (yet effective in practice) criterion depends on the degree of the boundary-defining polynomial $p \in \PP{i}$: the bound is approximated if $\dg{p}{i}$ exceeds a fixed threshold.
This has some limitations: (1) It can happen that \texttt{levelwise-scc} would not compute any (expensive) resultants with $p$, but this criterion would still advise inserting auxiliary polynomials.
While this defeats the purpose of avoiding expensive resultants, it is not completely useless as it still simplifies the cell description.
(2) Different thresholds will be optimal for different kinds of input problems, and it is not easy to guess a good one a priori.

Addressing the first issue, one can additionally check whether a resultant of $p$ with another nonlinear polynomial would be computed.
However, this did not improve the performance in our preliminary experiments.
The second issue is (partially) addressed with the dynamic termination criterion from \Cref{sec:incompleteness}, where the threshold for the degree grows with the number of approximated cells.

Instead of the degree $\dg{p}{i}$, we might also consider similar measures, e.g. the \emph{sum-of-total-degrees} of $p$'s monomials, which has been used for CAD projection orderings \cite{DBLP:conf/issac/DolzmannSS04}.

\subsubsection*{Transfer to CAlC.}
Similar to \nlsat, the \emph{cylindrical algebraic covering ({\footnotesize CAlC})} method \cite{AbrahamDEK21,covering} tries to extend a partial assignment $s \in \RR^i$ to a full model.
It uses the same theoretical framework as the levelwise construction to derive symbolic intervals to be excluded from the search.
In particular, it also ensures that certain root functions $\xi, \xi'$ do not intersect by making $\res{\xi.p}{\xi'.p}{i+1}$ order-invariant. 
We can apply our technique and dynamically introduce a polynomial with a root between $\xi,\xi'$, replacing an expensive resultant by two simpler ones, but shrinking the underlying cell.
This will introduce similar issues regarding termination, which can be solved in similar ways as presented above for SCC.

\section{Experiments}\label{sec:experiments}
We implemented several variants of our approach in the \smtrat solver \cite{corzilius_smt-rat_2015,noauthor_smt-rat_nodate}, allowing us to use its existing implementations of the levelwise SCC and \nlsat.
\smtrat also uses explanation backends based on the Fourier-Motzkin variable elimination (FM) \cite{FMSCC}, interval constraint propagation (ICP) \cite{KremerPhd}, and virtual substitution (VS) \cite{virtualS}, which are fast, but may fail to provide an explanation, especially for polynomials of degree 3 or higher.
The backends are called sequentially (FM, ICP, VS, SCC), so that the single cell construction is only needed when all other backends fail.
We compare the following variants for SCC:
\begin{description}
    \item[\texttt{Baseline}:] The original levelwise SCC.
    \item[\texttt{Simple-j}:] For $\texttt{j} \in \{3,4,5,6\}$, cell bounds defined by polynomials of degree \texttt{j} or higher are approximated by simple polynomials of the form $x_j - c$. Termination is ensured by limiting the number of approximated cells to $50$.
    \item[\texttt{Simple-*}:] a virtual best portfolio of the \texttt{Simple-j} variants.
    \item[\texttt{Dynamic}:] Uses the dynamic termination criterion from \Cref{sec:incompleteness}. A cell bound is approximated by a simple polynomial of the form $x_j - c$, if the defining polynomial has degree higher than or equal to $\sfrac{1}{5}\cdot n_{cells} + 3$, where $n_{cells}$ is the number of so far approximated cells.
    \item[\texttt{Taylor}:] Like \texttt{Dynamic}, but using Taylor approximations as in \Cref{sec:variants}.
    \item[\texttt{PWL-j}:] For $\texttt{j} \in \{2,4,6\}$, like \texttt{Dynamic}, but using piecewise linear approximations with \texttt{j} pieces, as presented in \Cref{sec:variants}.
    \item[\texttt{PWL-*}:] a virtual best portfolio of the \texttt{PWL-j} variants.
    \item[\texttt{Outside}:] If the lower (resp. upper) cell bound and another root function below (resp. above) that bound are defined by two non-linear polynomials, one of which fulfils the dynamic criterion, then we add a polynomial (of the form $x_j - c$) whose root lies outside the cell, between the bound and the other root, as explained in \Cref{sec:variants}.
\end{description}
Moreover, we also compare to the state-of-the-art \texttt{cvc5} \cite{cvc5} solver (version 1.2.1), which uses incremental linearization and cylindrical algebraic coverings.

We used the \qfnra benchmark set from \smtlib \cite{barrett_smt-lib_2010}, but only consider the $1684$ instances where \smtrat calls SCC at least once, since there is no difference between the variants on the other instances.
The tests were conducted on identical Intel\textregistered Xeon\textregistered 8468 Sapphire CPUs with 2.1 GHz per core, with a time limit of 1 minute and memory limit of 4GB per instance.

The results are summarized in \Cref{fig:results}.
While \texttt{cvc5} usually solves more instances than \smtrat on the entire \qfnra set, 
already the \texttt{Baseline} solver outperforms \texttt{cvc5} on our restricted benchmark set.

\paragraph*{Simple Approximations.}
Already the \texttt{Simple-j} variants solve around 45 instances more than \texttt{Baseline}.
However, these variants excel (partly) on different instances: the virtual portfolio \texttt{Simple-*} solves at least 18 instances more than each \texttt{Simple-j} variant.

\texttt{Dynamic} almost matches this performance, solving 58 instances more than \texttt{Baseline}, and it is in fact the best (non-portfolio) variant in our tests.
Interestingly, the differences in solved satisfiable instances and unsatisfiable instances are almost equal (30 more satisfiable, 28 more unsatisfiable).
This suggests that our explanations can not only help \nlsat find a model more quickly, but the Boolean structure of the unsatisfiable problems often still allows to cover the search space with the under-approximated cells.

\begin{figure}[h]
    \centering
    \includegraphics[scale=0.8]{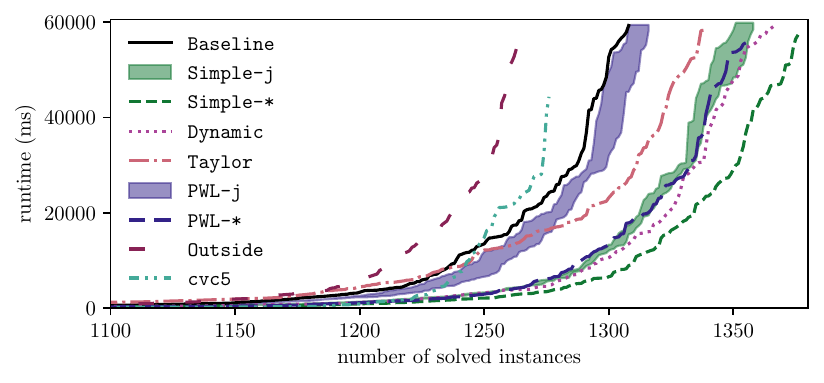}
    \caption{Performance profile. The purple area covers all \texttt{PWL-j} variants, the green area covers all \texttt{Simple-j} variants.}\label{fig:results}
\end{figure}

\Cref{fig:scatters} shows a more detailed comparison between \texttt{Dynamic} and \texttt{Baseline}.
Looking at the running times, \texttt{Dynamic} can solve many instances within one second, for which \texttt{Baseline} needs multiple seconds or even times out.
On the other hand, there are very few instances on which \texttt{Dynamic} is significantly slower, even though the number of SCC calls is usually much higher.
This is expected: the approximated cells can be computed faster, but they cover less of the search space.
Moreover, our approach often significantly reduces the maximum degree of any computed resultant.
However, we did not find a clear reduction of degrees of discriminants and coefficients.

\begin{figure}[t]
    \centering
    \includegraphics[scale=0.7]{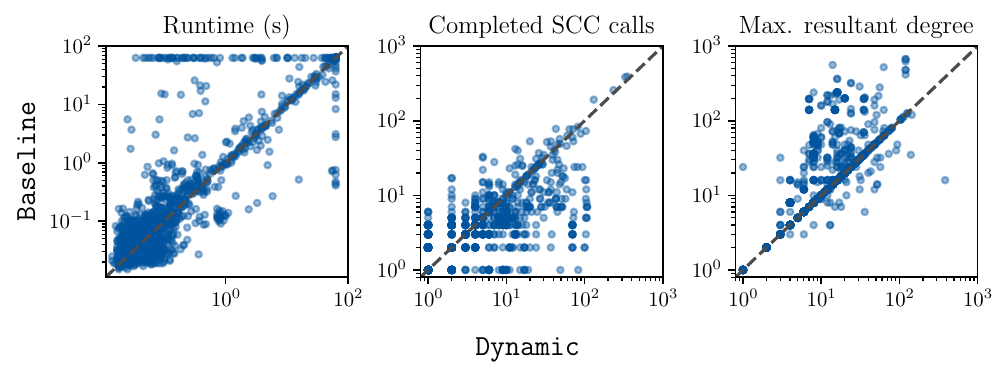}
    \caption{Scatter plots comparing \texttt{Dynamic} with \texttt{Baseline}}\label{fig:scatters}
\end{figure}

\paragraph*{Other Variants}
The \texttt{Taylor} and \texttt{PWL-j} variants performed slightly better than \texttt{Baseline}, but worse than the simple approximations.
It seems that the induced overhead outweighs potential benefits of better approximations.
In particular, our investigations often showed large bit sizes of the rational coefficients in the auxiliary polynomials and their resultants.
Another reason could be that the cell representation is less convenient, which is supported by the underwhelming performance of \texttt{Outside}, which solves fewer instances than \texttt{Baseline}.
While it does avoid some resultants like the other variants, it does not simplify the cell representation.
Accordingly, this seems to be a significant factor.

\paragraph*{Transfer to the CAlC Method.}
We also tested our modification in the context of the \emph{cylindrical algebraic covering (CAlC)} method \cite{covering}, which is implemented in \smtrat as well.
Now considering all 12154 instances of the QF\_NRA benchmark set, there is no significant difference between the baseline CAlC implementation (solving 9964 instances) and our modification (solving 9975 instances).
Once again, this might indicate that \nlsat especially benefits from simpler cell descriptions (which are not used in CAlC) and from compensating smaller cells by Boolean reasoning.

\paragraph*{A Note on Non-Termination.}
In \Cref{sec:incompleteness}, we showed that our approach can lead to non-termination with some sort of ``looping'' behaviour and presented ways to combat this.
Naturally, it would be interesting to know how often this occurs in practice.
As it is quite hard to reliably detect this behaviour, we have no concrete data.
However, when turning off the limit on approximated cells, the \texttt{Simple-j} variants time out significantly more often (solving around 40 instances less than with the limit), hinting at the importance of our termination criteria.
However, there are other influences in practice: When sampling new values, \smtrat prioritizes simple numbers, e.g. integers over rational numbers, and thus even the variants without the hard bound might (temporarily) escape looping when the gap between the actual cell boundary and the approximation becomes small. Then it may still time out in other parts of the computation or get caught in another loop.

\section{Conclusions}\label{sec:conclusion}

We modified the levelwise single cell construction for \nlsat and CAlC by dynamically inserting linear polynomials into the projection to avoid expensive high-degree polynomials in resultant computations and description of the resulting cell, at the cost of smaller cells - which are unterstood as ``under-approximations''. We introduced various variants to introduce such polynomials, as well as criteria to mitigate potential non-termination of \nlsat and CAlC.

In our experiments, our approach could significantly improve the running time of \nlsat. Interesingly, relatively simple under-aproximations performed best, while more complex approximations did not pay off, suggesting that the gains are mainly due to simpler root isolation in \nlsat.

There are several directions for further research:
Firstly, we conjecture that more intricate variants of \texttt{apx-criteria} may further improve efficiency.
This is not an easy task: In experiments not presented here, various other approaches based e.g. on the sum-of-total-degrees of the polynomials or involving the degrees of all polynomials on a level did not yield better results.
Secondly, our approach can only reduce computational effort for resultants, but not yet for discriminants, which have an even greater impact. 
Thirdly, the basic modification of the CAlC method had little impact, which might be improved. 
Finally, we guarantee termination of \nlsat by resorting to the original levelwise construction at some point.
Alternatively, we may consider $\delta$-completeness instead, as done e.g. in \cite{DREAL,ksmt}, and cover cells up to some precision $\delta$.

\subsubsection{Data Availability.}\label{data-avail}
Our implementation, experimental results, and tools for reproducing them are available at \url{https://doi.org/10.5281/zenodo.14916587}.

\begin{credits}
    \subsubsection{\ackname} J. Nalbach and V. Promies were supported by the DFG project SMT-ART (AB 461/9-1), and J. Nalbach and P. Wagner by the DFG RTG 2236 UnRAVeL.
    Computing resources were granted under the RWTH project rwth1560.
    
    \subsubsection{\discintname}
    The authors have no competing interests to declare that are relevant to the content of this article.
\end{credits}

\bibliography{references}
\end{document}